\documentclass[preprint,11pt, final, journal, letterpaper, twoside, onecolumn]{IEEEtran} 
\usepackage{mathptmx} 
\usepackage{times} 
\usepackage{amsmath} 
\usepackage{amssymb}  
\usepackage{enumerate}
\usepackage{mathrsfs}
\usepackage{psfrag}
\usepackage{amsbsy}
\usepackage{textcomp}
\usepackage{blkarray}
\usepackage{mathdots}
\usepackage[usenames,dvipsnames]{color}
\usepackage[dvips]{epsfig}
\usepackage{mdframed}
\usepackage{graphicx}
\usepackage{subcaption}
\usepackage{epsfig} 
\usepackage{algpseudocode,algorithm,algorithmicx}
\usepackage{tabularx}
\usepackage{multirow}
\usepackage{amsthm}

\newtheorem{theorem}{Theorem}[section]

\newtheorem{definition}[theorem]{Definition}
\newtheorem{corollary}[theorem]{Corollary}
\newtheorem{lemma}[theorem]{Lemma}

\newtheorem{example}[theorem]{Example}
\newtheorem{remark}[theorem]{Remark}
\newtheorem{assumption}[theorem]{Assumption}

\newcommand*\Let[2]{\State #1 $\gets$ #2}
\algrenewcommand\algorithmicrequire{\textbf{Precondition:}}
\algrenewcommand\algorithmicensure{\textbf{Postcondition:}}

\algnewcommand\Continue{\State\textbf{continue}}
\newcolumntype{Y}[1]{>{\hsize=#1\hsize}X}

\begin{document}
 \title{Impact of packet dropouts on the performance of optimal controllers and observers}
  \author{Amanpreet Singh Arora} \author{A. Sanand Amita Dilip \cormark[1]}
\author{Amanpreet Singh Arora, A. Sanand Amita Dilip \thanks{Amanpreet Singh Arora is with the Department of Aerospace Engineering, 
  Indian Institute of Technology Kharagpur, India. {\tt\small aps@iitkgp.ac.in}} \thanks{
   A. Sanand Amita Dilip is with the Department of Electrical Engineering, Indian Institute of Technology Kharagpur, India. 
  {\tt\small sanand@ee.iitkgp.ac.in}} 
 }
 \maketitle
\begin{abstract}
 We investigate  
the impact of packet dropouts due to non-idealities in communication networks on the performance of optimally derived controllers and observers in a minimax sense.
These packet dropouts are modeled by discrete constrained switching signals via directed graphs. 
We consider time optimal control and estimation, minimum energy and fuel optimization and LQR problems for systems subject to packet dropouts 
which turn out to be combinatorial; and compute 
algorithmically the worst case performance of these optimization problems. 
To reduce the computational burden in solving these combinatorial problems, a partial order is imposed on the switching signals. Validation of 
the proposed methods is done over a set of randomly generated systems. 
 Finally, by  
associating the worst case performance for each optimal control problem to the underlying network, one can optimize among the set of available communication 
networks. 
\end{abstract}

\section{Introduction}
%
With the rise of networked control and wireless communication in the last decade, classical models for control systems have been updated to incorporate 
the presence of non-idealities in the communication network 
(\cite{JungersKunduHeemels,schenato,mo,sin,heem,gom,paj,tab,hespanha,packetdropsiccps}
and the references therein). 
In many modern real life control systems, the plant and the 
controller are geographically distributed and communicate via wireless communication networks. Naturally, in real life situations, there is a possibility of a data loss 
due to packet dropouts in the communication network (\cite{JungersKunduHeemels,schenato,gom,cet,cet1}). Consequently, due to unreliable communication 
networks and loss of information, the control performance gets affected.  
Our goal is to quantify and compute the performance degradation due to the data loss constraint which allows us to measure the reliability of a 
communication network. 
To accommodate the data loss, we consider the following 
model proposed in \cite{JungersKunduHeemels}
\begin{eqnarray}\label{switch}
 \bold{x}(t+1) &=&\begin{cases} A\bold{x}(t)+B\bold{u}(t), \mbox{ if } \sigma(t)=1,\\
 A\bold{x}(t), \;\;\;\;\;\;\;\;\;\;\;\;\mbox{ if } \sigma(t)=0 \end{cases}\nonumber\\
 \bold{y}(t)&=&\begin{cases} C\bold{x}(t), \mbox{ if } \sigma(t)=1,\\
 \emptyset, \;\;\;\;\;\;\;\mbox{ if } \sigma(t)=0 \end{cases}
\end{eqnarray}
where $\sigma(t)=0$ when there is a packet dropout and $\sigma(t)=1$ when there is a successful transmission. We assume that there are certain communication 
constraints for example, no more than $k$ consecutive dropouts are allowed or during $n$ consecutive transmissions or there must be at least $m$ successful 
transmissions among every $n$ time instances and 
so on. These constraints are modeled using directed graphs (refer Preliminaries). Due to the presence of $\sigma$ in the model, 
instead of a linear time invariant system we have a switching system. 
A switching signal $\sigma$ is said to be admissible if it satisfies the given constraints. 
It is clear that for a fixed signal $\sigma$, the system model $(\ref{switch})$ is a linear time varying system model with $A(t)=A$ and $B(t)=B\sigma(t)$.  
For details about discrete linear systems under constrained switching, we refer the reader to \cite{Athanasopoulos2,Philippe} and the references therein. 
Schenato et al. \cite{schenato} consider probabilistic models to model the packet dropout phenomena. Kalman filter design for systems with data loss 
constraints is studied in \cite{mo,sin}. Constrained control with intermittent measurements was studied in \cite{rut} where the data loss is modeled 
using a language formed by a finite set alphabet. 
In \cite{Chakraborty}, optimal control problems due to feedback failure for continuous time LTI systems were studied. Recently, in \cite{savey}, intermittent 
feedback based control policy is proposed in the presence of bounded disturbance. 
Denial-of-service (DoS) attacks which destroy data availability via packet dropouts has also been studied intensively in the recent literature 
such as \cite{deper,feng,cet,cet1} and the references therein. 

We study some classical optimization problems in systems theory such as time optimal control and estimation and energy/fuel optimal control 
(\cite{Ath,liberzon}) for the model considered here and study the worst case performance over all admissible switching signals. 
Furthermore, we consider LQR optimization problems too for $(\ref{switch})$ and find the worst performance. We also consider the case where the designer 
finds an optimal LQR control law for an LTI system and then compute the performance degradation due to the data loss.  
To address these optimization problems, we use some results developed in \cite{packetdropsiccps} where  
the worst case minimum eigenvalue of the controllability Gramian was obtained over all admissible switching signals for the same model. 
To compute the worst minimum eigenvalue value of the controllability Gramian, an exact and an approximate approach was proposed in \cite{packetdropsiccps}. 
The exact approach involved a partial order on the admissible signals. Using this partial order, minimal admissible switching signals are constructed. 
We now extend this approach to tackle the optimization problems proposed in this paper  
and also leverage the well known optimization techniques of linear programming (LP) and semidefinite programming (SDP) 
(\cite{boyd,bentalnem}) in the proposed algorithms wherever applicable. We use the $T-$lift algorithm of \cite{packetdropsiccps} 
(details of $T-$lifts of graphs can be found in \cite{Philippe}) and the newly proposed 
algorithm in this paper for the generation of minimal signals and give a comparison of the two algorithms in solving the proposed combinatorial optimization problems. \\
{\bf Organization}: After stating preliminaries in the following section, we give a mathematical 
formulation of the optimization problems considered which are time optimal estimation and control, fuel and energy minimization, minimum reachability 
problem and LQR optimization in Section $\ref{pfm}$. In Section $\ref{topt}$, we give algorithms to solve time optimal control and estimation problems; whereas in Section $\ref{efopt}$, 
we solve fuel/energy optimization and minimum reachability problem exactly using minimal signals. Section $\ref{lqopt}$ involves LQR problems.
Section $\ref{heu}$ contains a new algorithm to generate 
minimal signals for a family of dropout signals where we 
reduce the amount of computations further. 
Validation study of the proposed methods is done in Section $\ref{study}$. \\
{\bf Notation}: Vectors are denoted by boldface letters, matrices by capital letters and scalars by small face letters. The symbol $\preceq$ denotes 
``less than or equal to'' in the given partial order. We use the commonly used short hand notation SDP and LP for semidefinite programs and linear 
programs respectively from the optimization literature. The set of real numbers and natural numbers is denoted by $\mathbb{R}$ and $\mathbb{N}$ respectively. 

\section{Preliminaries}\label{prel}
We first define some preliminaries to be used in the sequel. Many of these definitions are borrowed almost verbatim from \cite{packetdropsiccps}. 
\begin{definition}[\cite{packetdropsiccps}]\label{adm}
An automaton $\mathcal{A}$ is a directed labeled graph $G(V,E)$ with $N_V$ nodes in $V$ and $N_E$ edges in $E$. Each edge $(v,w,\sigma) \in E$ is labeled by $\sigma \in \{0,1\}$. A sequence $\sigma(0) \sigma(1),\ldots$ is \emph{admissible if there is a path in $G(V,E)$ carrying the sequence as the succession of labels 
   on its edges}. We denote by $\mathcal{L(A)}$ the set of all admissible switching sequences.
\end{definition}
\begin{example}[\cite{packetdropsiccps}]\label{ex1}
 Suppose no more than two consecutive dropouts are allowed. This is modeled using the following directed graph (refer Figure $1$). The set of admissible signals of a 
 fixed length $T$ can be found by enumerating all paths of length $T$ in the directed graph. 
 \begin{figure}[ht]\label{1drop}
\begin{center}
\includegraphics[scale=0.5]{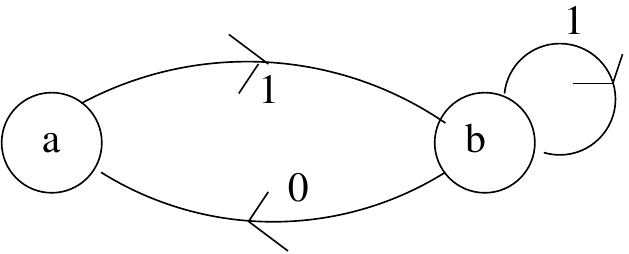}
\caption{{\scriptsize $G(V,E)$: No more than one consecutive dropout is allowed}}
\label{1dropnw}
\end{center}

\end{figure}
For example, when $T=4$, the set of admissible signals are $\{1111,1010,0101,0110,0111,1011\}$. 
\end{example}
We consider the system model $(\ref{switch})$  
where $\sigma(\cdot): \mathbb{N}\rightarrow \{ 0,1\}$ is a binary switching signal subject to constraints defined by an automaton which is a directed graph. 
The state evolution can be written as 
\begin{eqnarray}
  \bold{x}(t+1)=A^t\bold{x}(0)+C_{\sigma(t)}(A,B)\bar{\bold{u}}\label{lineq}
 \end{eqnarray}
 where  
\begin{eqnarray*}
 C_{\sigma(t)}(A,B):=
 \left[ \begin{array}{cccc} \sigma(0)A^{t}B  &\cdots  & \sigma(t-1)AB& \sigma(t)B\end{array}\right]
\end{eqnarray*}
is the controllability matrix associated with the switching sequence $\sigma(0)\sigma(1)...
\sigma(t)$ at time $t$ and $\bar{\bold{u}}:=\left[ \begin{array}{ccccc}  \bold{u}(0)^{\mathsf{T}}&\cdots& \bold{u}(t)^{\mathsf{T}}\end{array}\right]^{\mathsf{T}}$. One can define 
the observability matrix in a similar manner as 
\begin{eqnarray*}
 O_{\sigma(t)}(C,A):=
 \left[ \begin{array}{cccc} \sigma(0)C^{\mathsf{T}}  \sigma(1)(CA)^{\mathsf{T}}&\cdots  &  \sigma(t)(CA^t)^{\mathsf{T}}\end{array}\right]^{\mathsf{T}}.
\end{eqnarray*}
\begin{definition}[\cite{JungersKunduHeemels,packetdropsiccps}]\label{cntrl}
 The system $(\ref{switch})$ is reachable, if for any 
$\sigma \in \mathcal{L(A)}$, any final state $\bold{x}_f \in \mathbb{R}^n$ and for $\bold{x}_0=\bold{0}$, there exists an input signal $\bold{u}(\cdot) : \mathbb{N} \rightarrow \mathbb{R}^m$ 
such that $\bold{x}(t_{\sigma}) = \bold{x}_f$ for some $t_{\sigma} \in \mathbb{N}$. 
If, furthermore, the property holds for all $\bold{x}_0 \in \mathbb{R}^n$, we say that 
$(\ref{switch})$ is controllable and if $\bold{x}_0\in \mathbb{R}^n$ and $\bold{x}_f=\bold{0}$ in the above case, we say that $(\ref{switch})$ is $0-$controllable. 
\end{definition}
\begin{definition}\cite{JungersKunduHeemels}
 The system $(\ref{switch})$ is observable, if for any $\sigma \in \mathcal{L(A)}$, any pair of initial states $\bold{x}_0,\tilde{\bold{x}}_0\in \mathbb{R}^n$, 
 it holds that $\bold{y}(t,\bold{x}_0,\sigma)=\bold{y}(t,\tilde{\bold{x}}_0,\sigma)$ for all $t\in \mathbb{N}$ implies that $\bold{x}_0=\tilde{\bold{x}}_0$. 
 The system is called unobservable otherwise.
\end{definition}

Controllability and observability properties of $(\ref{switch})$ can be checked by the rank of corresponding controllability and observability matrices 
$C_{\sigma(t)}(A,B)$ and $ O_{\sigma(t)}(C,A)$ respectively over all admissible switching signals (\cite{JungersKunduHeemels}). It was shown by   
Jungers, Kundu and Heemels (\cite{JungersKunduHeemels}) that controllability and observability of these models can be decided in polynomial time 
by giving an explicit algorithm. 
We now define a partial order on the set of admissible signals which will be used to reduce the amount of computations required to solve the proposed 
optimization problems. 
\begin{definition}[\cite{packetdropsiccps}]\label{ord}
Given two switching signals $\sigma_1$ and $\sigma_2$, we write $\sigma_1 \preceq \sigma_2$ if for all $ i\in \mathbb{Z}$ for which 
$\sigma_1(i)=1$ it follows that $\sigma_2(i)=1$. 
\end{definition}
\begin{definition}[\cite{packetdropsiccps}]\label{irred}
 We say that a signal $\sigma$ is minimal, if there does not exist any other admissible signal $\bar{\sigma}$ $(\bar{\sigma}\neq \sigma)$ 
 such that $\bar{\sigma} \preceq \sigma$. We denote by $M_{\sigma(t)}$ the set of all minimal signals of length $t+1$ (since the time index starts from $0$). 
\end{definition}
\begin{example}
 In Example $\ref{ex1}$, let $\sigma_1=\{0101\}$ and $\sigma_2=\{0111\}$. Then, $\sigma_1\preceq \sigma_2$. The set of minimal signals for Example $\ref{ex1}$ 
 is $M_{\sigma(3)}=\{1010,0101,0110\}$.
\end{example}

We refer the reader to \cite{packetdropsiccps} (Algorithm $1$) for the algorithm to generate minimal signals of a fixed length. 
The minimal signals are constructed iteratively using the method of $T-$lift of graphs (\cite{packetdropsiccps}). 
\begin{assumption}\label{assum1}
 We assume that $A$ is invertible.
\end{assumption}
We refer the reader to \cite{JungersKunduHeemels} for the controllability and observability properties of $(\ref{switch})$ when $A$ is non invertible. For 
the sake of simplicity, we avoid this case by Assumption $\ref{assum1}$. 
\section{Problem formulation}\label{pfm}
We assume throughout that $(\ref{switch})$ is controllable and observable. 
We now mention the combinatorial problems to be considered.
\newcounter{problems}
\setcounter{problems}{1}
\begin{enumerate}
\item Problem \Roman{problems} (Worst time optimal estimation): 
Given a fixed time $T$, determine if it is possible to estimate $\bold{x}(0)$ in $T$ time steps for the control system $(\ref{switch})$. If yes, then 
find the worst optimal time required for the estimation of $\bold{x}(0)$ over all admissible switching signals i.e.,  
 \begin{eqnarray}
  &\mbox{max}_{\sigma\in\mathcal{L}(\mathcal{A})}\mbox{min}&t\nonumber\\
 &\mbox{subject to}& \mbox{rank}(O_{\sigma(t)}(C,A))=n.\label{west}
 \end{eqnarray}
\addtocounter{problems}{1}
 \item Problem \Roman{problems} (Worst time optimal control): Suppose that the control system $(\ref{switch})$ is $0-$controllable at time $T$ with input constraints $|u_i(t)|\le1$. 
 Find the worst optimal time required to drive the state to the origin from an arbitrary initial condition 
 $\bold{x}(0)\in \mathbb{R}^n$ over all admissible switching signals. 
 \begin{eqnarray}
&\mbox{max}_{\sigma\in\mathcal{L}(\mathcal{A})}\mbox{min}_{\bar{\bold{u}}}&t\nonumber\\
 &\mbox{subject to}& \bold{0}=A^{t}\bold{x}(0)+C_{\sigma(t)}(A,B)\bar{\bold{u}},\nonumber\\
 &&\bold{x}(0)\in\mathbb{R}^n, \;|u_i(t)|\le1.\label{wtime}
\end{eqnarray}
\addtocounter{problems}{1}
\item Problem \Roman{problems} (Worst fuel and energy optimal control): Consider the optimization problem 
\begin{eqnarray}
 &\mbox{max}_{\sigma\in\mathcal{L}(\mathcal{A})}\mbox{min}_{\bar{\bold{u}}}&\gamma_1\|\bar{\bold{u}}\|_1+\gamma_2\|\bar{\bold{u}}\|_2\nonumber\\
 &\mbox{subject to}& \bold{x}_{f}=C_{\sigma(t)}(A,B)\bar{\bold{u}}.\label{wminfuelen}
\end{eqnarray}
\addtocounter{problems}{1}
It is clear that when $\gamma_1=0,\gamma_2\neq0$, one obtains the max-min energy optimization problem and when 
$\gamma_1\neq0,\gamma_2=0$, one obtains the max-min fuel optimization problem. 
\item Problem \Roman{problems} (Minimum reachability): Suppose the available input energy is less than or equal to one. Let $\mathcal{P}_1$ be a polytope in the state space defined by 
conv$\{\bold{v}_1,\ldots,\bold{v}_p\}$. 
Verify if all points in the polytope are reachable from the origin under all switching 
signals with the unit energy constraint for the control system $(\ref{switch})$. 
\addtocounter{problems}{1}
 \item  Problem \Roman{problems} (Maxmin LQR): 
 Consider  
\begin{eqnarray}
 &\mbox{min}_{\bold{u}(t)}& \bold{x}^{\mathsf{T}}(T)Q_f\bold{x}(T)+\sum_{t=0}^{T-1}\bold{x}^{\mathsf{T}}(t)Q\bold{x}(t)+
 \bold{u}^{\mathsf{T}}(t)R\bold{u}(t)\nonumber\\
 &\mbox{subject to}& \bold{x}(t+1)=A\bold{x}(t)+B\sigma(t)\bold{u}(t),\nonumber\\
  &&\sigma \in \mathcal{L}(\mathcal{A}) .\label{lqr}
\end{eqnarray}
For a fixed $\sigma$, $\bold{x}(t+1)=A\bold{x}(t)+B\sigma(t)\bold{u}(t)$ forms a linear time varying system and the optimal cost is given by 
$\bold{x}^{\mathsf{T}}(0)P_{\sigma}(0)\bold{x}(0)$ where $P_{\sigma}(t)$ satisfies the difference Riccati equation 
$P_{\sigma}(t)=Q+A^{\mathsf{T}}P_{\sigma}(t+1)A-\sigma(t)A^{\mathsf{T}}P_{\sigma}(t+1)B(R+B^{\mathsf{T}}
  P_{\sigma}(t+1)B)^{-1}B^{\mathsf{T}}P_{\sigma}(t+1)$, $P_{\sigma}(T)=Q_f$ for each $\sigma\in \mathcal{L}(\mathcal{A})$. 
  We want to find the worst LQR energy which can be obtained by solving 
\begin{eqnarray}
 &\mbox{max}_{\sigma\in\mathcal{L}(\mathcal{A})}& \bold{x}^{\mathsf{T}}(0)P_{\sigma(0)}\bold{x}(0)\nonumber\\
 &\mbox{subject to}& P_{\sigma}(t)=Q+A^{\mathsf{T}}P_{\sigma}(t+1)A-\sigma(t)A^{\mathsf{T}}
  P_{\sigma}(t+1)B(R+B^{\mathsf{T}}P_{\sigma}(t+1)B)^{-1}B^{\mathsf{T}}P_{\sigma}(t+1),\nonumber\\
  &&P_{\sigma}(T)=Q_f.\label{wlqr}
\end{eqnarray}
\addtocounter{problems}{1}
\item Problem \Roman{problems} (Fixed input LQR): Consider the cost function 
\begin{equation}
 \mbox{min}_{\bold{u}(t)}\; \bold{x}^{\mathsf{T}}(T)Q_f\bold{x}(T)+\sum_{t=0}^{T-1}\bold{x}^{\mathsf{T}}(t)Q\bold{x}(t)+
 \bold{u}^{\mathsf{T}}(t)R\bold{u}(t)\nonumber
\end{equation}
which is the cost function for the finite horizon LQR problem. 
Let $\bold{u}(t)=K(t)\bold{x}(t)$ be the optimal state feedback law obtained for the LTI system $\bold{x}(t+1)=A\bold{x}(t)+B\bold{u}(t)$ without any dropout constraints  
where \emph{we assume that states are available for feedback and have been reconstructed 
ideally by the controller}. Therefore, the closed loop dynamics are $\bold{x}(t+1)=(A+BK(t))\bold{x}(t)$. In other words, the control law is obtained for 
an ideal system. 
Suppose there are packet dropouts in the closed loop system leading to $\bold{x}(t+1)=(A+\sigma(t)BK(t))\bold{x}(t)$ where $\sigma$ is a constrained 
switching signal. 
We want to find the worst degraded cost 
\begin{eqnarray}
 \mbox{max}_{\sigma} \{\bold{x}^{\mathsf{T}}(T)Q_f\bold{x}(T)+\sum_{t=0}^{T-1}\bold{x}^{\mathsf{T}}(t)(Q+
 K^{\mathsf{T}}(t)RK(t))\bold{x}(t)\}.\nonumber
\end{eqnarray}
\addtocounter{problems}{1}
\end{enumerate}
\begin{remark}
 The solutions of the optimization problems defined above (except Problem $IV$) can be used to assign performance degradation measures for the underlying communication network. 
If multiple communication networks are available, then one can optimize among these networks based on the performance measures depending on the type of
optimal control problems one wants to solve. In other words, the solutions to these problems can be used as a measure to select the optimal 
network.   
\end{remark}

\section{Time optimal control and estimation}\label{topt} 
In this section, we algorithmically solve the combinatorial optimization problems $I$ and $II$ given by Equations 
$(\ref{west})$ and $(\ref{wtime})$ in the previous section. 
\subsection{Time optimal estimation} 
We give the following lemma which allows us to find an algorithmic solution to Problem $I$. 
  \begin{lemma}
Suppose $(\ref{switch})$ is observable. 
Let $\sigma_1,\sigma_2\in \mathcal{L}(\mathcal{A})$.  
Let $t_{\sigma_1},t_{\sigma_2}$ be the time instance such that 
$O_{\sigma(t_{\sigma_1})}$ and $O_{\sigma(t_{\sigma_2})}$ become full column rank for the first time. 
If $\sigma_1\preceq \sigma_2$, then $t_{\sigma_1}\ge t_{\sigma_2}$. 
\end{lemma}
\begin{proof}
 Follows from Definition $\ref{ord}$ and the definition of the observability matrix. 
\end{proof}

By the above lemma, it is enough to consider the rank of the observability matrices 
associated with the minimal signals instead of all admissible signals. Algorithm $1$  
 illustrates how to solve Problem $I$. 
\begin{algorithm}\label{alg:worstopttime1}
 \caption{Finding the worst optimal time for state estimation 
    }
    \hspace*{\algorithmicindent} \textbf{Input : $G(V,E), T,\{\bold{y}(0),\ldots,\bold{y}(T)\},A,C$}\\
    \hspace*{\algorithmicindent} \textbf{Output : Worst estimation time $t_{w}$}
    \begin{algorithmic}[1]
    \Let{$M_{\sigma(T)}$}{\mbox{set of minimal signals of length $T$}}
     \For{every $\sigma\in M_{\sigma(T)}$}
     \Let{$t_{\sigma}$}{\mbox{least $t\in \mathbb{N}, t\le T$ such that rank$(O_{\sigma(t)})=n$}}
     \If {rank$(O_{\sigma(T)})<n$}
     \Let{$t_{\sigma}$}{$\infty$}
     \EndIf
     \EndFor
     \Let{$t_w$}{\mbox{max}$_{\sigma\in M_{\sigma(T)}}\{t_{\sigma}\}$}
     \State \Return{$t_{w}$}
    \end{algorithmic}

\end{algorithm}
\begin{remark}
Notice that, if 
$t_w\le T$, then the estimation problem is feasible otherwise it is infeasible. 
It follows that the algorithm terminates since the number of minimal signals is finite and $T$ is also finite. 
\end{remark}

\subsection{Time optimal control}
To solve Problem $II$, 
we prove the following result to show that instead of all admissible signals, one only needs the set of minimal signals 
to find the solution. 
\begin{theorem}
 Suppose $(\ref{switch})$ is controllable and the constraint in Equation $(\ref{wtime})$ is feasible. 
 Let $\sigma_1, \sigma_2\in \mathcal{L}(\mathcal{A})$ such that $\sigma_1\preceq \sigma_2$ and let $t_{\sigma_1}, t_{\sigma_2}$ be the optimal time required to drive the state to 
 the origin from an initial condition $\bold{x}(0)$ for the two switching 
 signals $\sigma_1$ and $\sigma_2$ respectively. Then $t_{\sigma_1}\ge t_{\sigma_2}$. 
\end{theorem}
\begin{proof}
  From the hypothesis,
  \begin{equation}
   -A^{t_{\sigma_1}}\bold{x}(0)=C_{\sigma_1(t_{\sigma_1})}(A,B)\bold{u}_{\sigma_1}.\nonumber
  \end{equation}
 Since $\sigma_1\preceq \sigma_2$, it follows that 
 \begin{equation}
  -A^{t_{\sigma_1}}\bold{x}(0)=C_{\sigma_2(t_{\sigma_1})}(A,B)\bold{u}_{\sigma_1}.\nonumber
 \end{equation}
 Now by definition, $t_{\sigma_2}$ is the optimal time to reach $x_f$ from the 
 origin for the switching signal $\sigma_2$. Therefore, $t_{\sigma_1}\ge t_{\sigma_2}$. 
\end{proof}
It is clear from the above theorem that to solve the optimization problem given by Equation $(\ref{wtime})$, it is enough to consider the 
set of minimal signals (Definition $\ref{ord}$). 
Now assuming that $(\ref{switch})$ is $0-$controllable in time $T$ and feasibility of the constraint in Equation $(\ref{wtime})$, 
the following algorithm (Algorithm $2$) solves Problem $II$. 
\begin{algorithm}
 \caption{Finding the worst optimal time for a state transfer to the origin from an arbitrary initial condition 
    \label{alg:worstopttime}}
    \hspace*{\algorithmicindent} \textbf{Input : $G(V,E),A,B, \bold{x}_0,T$}\\
    \hspace*{\algorithmicindent} \textbf{Output : Worst optimal time $t_{w}$}
    \begin{algorithmic}[1]
    \Let{$M_{\sigma(T)}$}{\mbox{the set of minimal signals of length $T$}}
     \For{every $\sigma\in M_{\sigma(T)}$}
     \Let{$t_{\sigma}$}{\mbox{min$_{t\in \mathbb{N}}\{-A^{t_{\sigma}}\bold{x}_0=C_{\sigma(t_{\sigma})}(A,B)\bar{\bold{u}}, 
     \|\bar{\bold{u}}\|_{\infty}\le1\}$}}
     \EndFor
     \Let{$t_w$}{\mbox{max}$_{\sigma\in M_{\sigma(T)}}\{t_{\sigma}\}$}
     \State \Return{$t_{w}$}
    \end{algorithmic}

\end{algorithm}
\begin{remark}
 Notice that due to our feasibility assumption on $\bold{\bar{u}}$ and the finiteness of the minimal signals, Algorithm $2$ always terminates. Without 
 the feasibility assumption, we may have a situation where $\bold{\bar{u}}$ does not exists for a particular switching signal $\sigma$. Such cases can 
 be incorporated in Algorithm $2$ by making $t_{\sigma}=\infty$ to include infeasible cases.
\end{remark}


\section{Fuel and energy optimization and minimum reachability}\label{efopt}
\subsection{Fuel and energy optimization}
In the control literature, minimizing the $2-$norm of the input (or the $l_2-$norm) in Equation $(\ref{lineq})$ is identified with minimum energy problems whereas; minimizing the 
$1-$norm (or the $l_1-$norm) is identified with minimum fuel problems. 
We show in the following result that the search space can be reduced to the set of minimal signals to solve 
$(\ref{wminfuelen})$. 
\begin{theorem}
 Suppose $(\ref{switch})$ is controllable. Let $\sigma_1, \sigma_2 \in \mathcal{L}(\mathcal{A})$. Then, 
 \begin{eqnarray}
  &&\sigma_1\preceq \sigma_2\Rightarrow \|\bar{\bold{u}}_{\sigma_2}\|_1\le \|\bar{\bold{u}}_{\sigma_1}\|_1\mbox{ and}\nonumber\\
  &&\gamma_1\|\bar{\bold{u}}_{\sigma_2}\|_1
 +\gamma_2\|\bar{\bold{u}}_{\sigma_2}\|_2
 \le \gamma_1\|\bar{\bold{u}}_{\sigma_1}\|_1+\gamma_2\|\bar{\bold{u}}_{\sigma_1}\|_2.\nonumber
 \end{eqnarray}
  
\end{theorem}
\begin{proof}
 Let $\bar{\bold{u}}_{\sigma_1}$ and $\bar{\bold{u}}_{\sigma_2}$ be the $1-$norm minimizing solutions of $(\ref{minfuel})$ for switching signals $\sigma_1$ and $\sigma_2$ respectively. 
 Since $\sigma_1\preceq \sigma_2$, $\bold{x}_f= C_{\sigma_2(t)}(A,B)\bar{\bold{u}}_{\sigma_1}$. Now because $\bar{\bold{u}}_{\sigma_2}$ is a norm minimizing 
 solution, $\|\bar{\bold{u}}_{\sigma_2}\|_1\le \|\bar{\bold{u}}_{\sigma_1}\|_1$. The other inequality follows similarly. 
\end{proof}

Consider the following optimization problems for a fixed switching signal $\sigma$ 
and fixed time $t$. 
\begin{eqnarray}
 &\mbox{min}&\|\bar{\bold{u}}\|_1\nonumber\\
 &\mbox{subject to}& \bold{x}_f=
 C_{\sigma(t)}(A,B)\bar{\bold{u}}\label{minfuel}\\
 &\mbox{min}&\gamma_1\|\bar{\bold{u}}\|_1+\gamma_2\|\bar{\bold{u}}\|_2\nonumber\\
 &\mbox{subject to}& \bold{x}_f=
 C_{\sigma(t)}(A,B)\bar{\bold{u}}.\label{minfuelen}
\end{eqnarray}
Note that the optimization problem given by $(\ref{minfuel})$ can be converted into an LP. 
The optimization problem $(\ref{minfuel})$
with input constraints $|u_i(j)|\le1$, $j=0,\ldots,t$ can be rewritten as (\cite{boydvan}, p.294)
\begin{eqnarray}
 &\mbox{min}_{\bold{n}}&\bold1^{\mathsf{T}}\bold n\nonumber\\
 &\mbox{subject to}& \bold{x}_f=
 C_{\sigma(t)}(A,B)\bar{\bold{u}},\;|u_i(j)|\le n_{ij},\nonumber\\
 && \;n_{ij} \le 1,\;i=1,\ldots,m, j=0,\ldots,t, \label{minfuelcons}
\end{eqnarray}
where $\bold{n}=\left[ \begin{array}{cccc} n_{10}&\cdots& n_{mt}\end{array}\right]^{\mathsf{T}}$ 
which is equivalent to  
\begin{eqnarray}
 &\mbox{min}_{\bold{n}}&\bold1^{\mathsf{T}}\bold n\nonumber\\
 &\mbox{subject to}& \bold{x}_f=
 C_{\sigma(t)}(A,B)\bar{\bold{u}}\nonumber\\
 &&\bar{\bold{u}}\le\bold{n},\;\bar{\bold{u}}\ge-\bold{n},\bold{n}\le\bold{1}. \label{minfuelcons1}
\end{eqnarray}
This is an LP. 
One needs to solve this LP over the set of minimal signals and find the solution with the maximum $1-$norm. 
Algorithm $\ref{alg:worstoptfuel}$ provides a method to find the worst optimal fuel input over all switching signals. 
\begin{algorithm}
 \caption{Finding the worst optimal fuel input (worst $\|.\|_1$ input) for a state transfer to the origin from an arbitrary initial condition 
    \label{alg:worstoptfuel}}
    \hspace*{\algorithmicindent} \textbf{Input : $G(V,E),A,B, \bold{x}_0,T$}\\
    \hspace*{\algorithmicindent} \textbf{Output : Worst optimal fuel input}
    \begin{algorithmic}[1]
    \Let{$M_{\sigma(T)}$}{\mbox{set of minimal signals of length $T$ }}
         \For{every $\sigma\in M_{\sigma(T)}$}
     \Let{$\bold{u}_{\sigma}$}{\mbox{$\bar{\bold{u}}$ such that $\bar{\bold{u}}$ is a solution of $(\ref{minfuelcons1})$}}
     \EndFor
     \Let{$\bold{u}_w$}{\mbox{max}$_{\sigma\in M_{\sigma(T)}}\{\bold{u}_{\sigma}\}$}
     \State \Return{$\bold{u}_{w}$}
    \end{algorithmic}

\end{algorithm}
\begin{remark}
 A similar algorithm can be constructed to find the worst case input of $(\ref{minfuelen})$. 
 If $\gamma_1=0$ in Problem $(\ref{minfuelen})$, then for a fixed $\sigma$, the problem becomes 
 an SDP. For $\gamma_1\neq0$, this is SOCP which can also be formulated as an SDP (\cite{boydvan}, p.310).  Thus, we need to solve SDPs over the set of minimal 
 signals to solve the optimization problem. 
\end{remark}

\subsection{Minimum reachability set with energy bounds}
We now give the following lemma to address Problem $IV$. 
\begin{lemma}
Suppose $(\ref{switch})$ is controllable. Let $\sigma_1, \sigma_2 \in \mathcal{L}(\mathcal{A})$ such that 
 $\sigma_1\preceq\sigma_2$. Then, 
 $\mathcal{E}_1(t):=\{\bold{x}\in\mathbb{R}^n\;|\; \bold{x}^{\mathsf{T}}W_{\sigma_1(t)}^{-1}\bold{x}\le1\}$ 
 $\subseteq$ $\mathcal{E}_2(t):=\{\bold{x}\in\mathbb{R}^n\;|\; \bold{x}^{\mathsf{T}}W_{\sigma_2(t)}^{-1}\bold{x}\le1\}$.  
\end{lemma}
\begin{proof}
 Note that if $\sigma_1\preceq\sigma_2$, then $W_{\sigma_1(t)}\le W_{\sigma_2(t)}$. Therefore, $W_{\sigma_1(t)}^{-1}\ge W_{\sigma_2(t)}^{-1}$. Thus, 
 $\bold{x}^{\mathsf{T}}W_{\sigma_1(t)}^{-1}\bold{x}\le1\Rightarrow \bold{x}^{\mathsf{T}}W_{\sigma_2(t)}^{-1}\bold{x}\le1$ i.e.,  $\mathcal{E}_1(t)\subseteq 
 \mathcal{E}_2(t)$.
\end{proof}
\begin{remark}
 It follows that if $\mathcal{P}_1\in \cap_{\sigma \in  M_{\sigma}(T)} \mathcal{E}_{\sigma}(T)$ where $M_{\sigma}(T)$ is the set of minimal signals, then $\mathcal{P}_1$ is reachable 
in time $T$ for an arbitrary packet dropout signal. An ellipsoid $\mathcal{E}_{\sigma}(T)$ contains $\mathcal{P}_1$ $\Leftrightarrow$ 
$\bold{v}_j^{\mathsf{T}}W_{\sigma(T)}^{-1}\bold{v}_j\le1$, $j=1,\ldots,p$ (\cite{boyd}). Thus, $\mathcal{P}_1$ is reachable under arbitrary constrained 
switching signal if and only if $\bold{v}_j^{\mathsf{T}}W_{\sigma(T)}^{-1}\bold{v}_j\le1$, $j=1,\ldots,p$ for all $\sigma \in M_{\sigma(T)}$.
\end{remark}

\section{LQR Problems}\label{lqopt}
In this section, we consider two types of problems namely the max-min LQR problem and the worst case LQR for the fixed input. 
\subsection{Maxmin LQR}
For a fixed switching signal $\sigma$, the LQR problem given by Equation $(\ref{lqr})$ is exactly the same as 
the classical linear time varying systems problem with $A(t)=A,B(t)=B\sigma(t)$. 
We get Riccati equations with time varying coefficients corresponding to $B\sigma(t)$ whose solution determines the feedback law. We use the idea of the minimal signals to find the worst energy over all switching signals. 
The worst case LQR measure can be used for closed loop control problems which forms a counterpart of the open loop control using the controllability Gramian. 
\begin{theorem}
 Let $\sigma_1\preceq \sigma_2$ and $P_{\sigma_1}(t),P_{\sigma_2}(t)$ be the matrices for $\sigma_1,\sigma_2$ satisfying the difference Riccati equations 
\begin{eqnarray}
  P_{\sigma_1}(t)=Q+A^{\mathsf{T}}P_{\sigma_1}(t+1)A-
  \sigma_1(t)A^{\mathsf{T}}P_{\sigma_1}(t+1)B(R+B^{\mathsf{T}}P_{\sigma_1}(t+1)B)^{-1}B^{\mathsf{T}}P_{\sigma_1}(t+1)\label{dric1}\\
  P_{\sigma_2}(t)=Q+A^{\mathsf{T}}P_{\sigma_2}(t+1)A-
  \sigma_2(t)A^{\mathsf{T}}P_{\sigma_2}(t+1)B(R+B^{\mathsf{T}}P_{\sigma_2}(t+1)B)^{-1}B^{\mathsf{T}}P_{\sigma_2}(t+1)\label{dric2}
 \end{eqnarray}
 with the boundary condition $P_{\sigma_1}(t_f)=P_{\sigma_2}(t_f)=Q_f>0$. Then, $P_{\sigma_1}(0)\ge P_{\sigma_2}(0)$. 
\end{theorem}
\begin{proof}
 Let $P_{\sigma_1}(t)$ and $P_{\sigma_2}(t)$ be the corresponding matrices satisfying the two difference Riccati equations. The total cost is 
 $\bold{x}_0^{\mathsf{T}}P_{\sigma_1}(0)\bold{x}_0$ and $\bold{x}_0^{\mathsf{T}}P_{\sigma_1}(0)\bold{x}_0$ respectively and $P_{\sigma_1}(t_f)=P_{\sigma_2}(t_f)=Q_f>0$. 
 Let $\sigma_1\preceq \sigma_2$, then 
 \begin{eqnarray}
  P_{\sigma_1}(t_f-1)&=&Q+A^{\mathsf{T}}Q_fA-
  \sigma_1(t_f)A^{\mathsf{T}}Q_fB(R+B^{\mathsf{T}}Q_fB)^{-1}B^{\mathsf{T}}Q_f\label{dric3}\\
  P_{\sigma_2}(t_f-1)&=&Q+A^{\mathsf{T}}Q_fA-
  \sigma_2(t_f)A^{\mathsf{T}}Q_fB(R+B^{\mathsf{T}}Q_fB)^{-1}B^{\mathsf{T}}Q_f.\label{dric4}
 \end{eqnarray}
 Observe that since $A^{\mathsf{T}}Q_fB(R+B^{\mathsf{T}}Q_fB)^{-1}B^{\mathsf{T}}Q_f\ge 0$, and $\sigma_1\preceq \sigma_2$, it is clear that 
 $P_{\sigma_1}(t_f-1)\ge P_{\sigma_2}(t_f-1)$. By similar arguments, going backwards in time, 
 $P_{\sigma_1}(0)\ge P_{\sigma_2}(0)$. 
\end{proof}

\subsection{Fixed input LQR}
As mentioned in the Problem $VI$, let $\bold{u}=K\bold{x}(t)$ be the optimal input obtained for the LTI system 
$\bold{x}(t+1)=A\bold{x}(t)+B\bold{u}(t)$ for the classical finite horizon LQR problem. We 
now show that the worst case performance degradation can be computed using the partial order relation. 
\begin{theorem}
 Let $J_{\sigma_1}$ and $J_{\sigma_2}$ be the LQR cost associated with switching signals $\sigma_1$ and $\sigma_2$ when $\bold{u}=K(t)\bold{x}(t)$ is used 
 as a stabilizing state feedback law. Suppose $A$ is unstable but $(A,B)$ is stabilizable. 
 If $\sigma_1 \preceq \sigma_2$, then $J_{\sigma_2}\preceq J_{\sigma_1}$. 
\end{theorem}
\begin{proof}
 Observe that the cost is given by 
 \begin{eqnarray}
  \bold{x}^{\mathsf{T}}(T)Q_f\bold{x}(T)+\sum_{t=0}^{T-1}\bold{x}^{\mathsf{T}}(t)(Q+
 K^{\mathsf{T}}(t)RK(t))\bold{x}(t)\nonumber\\
 =\bold{x}^{\mathsf{T}}(T)Q_f\bold{x}(T)+\bold{x}^{\mathsf{T}}(0)\Bigg(\sum_{t=0}^{T-1}((A+\sigma(t)BK(t))^{\mathsf{T}})^{t}
 (Q+ K^{\mathsf{T}}(t)RK(t))(A+\sigma(t)BK(t))^{t}\Bigg)\bold{x}(0).\nonumber
 \end{eqnarray}
 When $\sigma(t)=1$, $(A+BK(t))$ is a Schur 
 matrix which implies that it acts as a contraction map on the state space. Therefore, 
 $\bold{x}^{\mathsf{T}}(0)((A+BK(t))^{\mathsf{T}})^{t}Q(A+BK(t))^t\bold{x}(0)
 \le \bold{x}^{\mathsf{T}}(0)(A^{\mathsf{T}})^tQA^t\bold{x}(0)$. Therefore, $\sigma_1 \preceq \sigma_2 \Rightarrow J_{\sigma_2}\preceq J_{\sigma_1}$. 
\end{proof}


\section{Minimal Signal Generation}\label{heu} 
%
In \cite{packetdropsiccps}, the $T-$lift algorithm for generating minimal signals of switching signals modeled by an automaton was proposed 
using ideas introduced in \cite{Philippe}. While the $T-$lift algorithm works well for the case of general automaton, for the case of signals with at most $k$ consecutive dropouts, a faster algorithm can be formulated. We now propose this faster algorithm and justify why it works.
\begin{definition}[Admissible Flipping]
A flipping of a bit of a binary switching signal (i.e., changing a bit from $1$ to $0$ or vice versa)
is said to be an admissible flipping if and only if the flipping results in an admissible signal. Moreover, the direction of the flip can also be specified 
while describing admissibility, for example, an Admissible Flip $0\rightarrow 1$ means flipping of a $0$ bit to $1$ is admissible and vice versa for an Admissible Flip $1\rightarrow 0$.
\end{definition}
\begin{definition}[Surrounding Bits]
If a $1$ bit of a binary switching signal is preceded by $p$ and succeeded by $q$ consecutive $0$ bits then such a $1$ bit is said to be surrounded by $p+q$ number of $0$ bits. The same definition applies for a $0$ bit surrounded by $1$ bits. 
\end{definition}
\begin{theorem}\label{flip_theo}
A switching signal is minimal if and only if it is admissible and no combination of $1\rightarrow 0$ flippings is admissible
\end{theorem}
\begin{proof}
For a switching signal, a weaker signal in the sense of partial order defined earlier can only be constructed by $1\rightarrow 0$ flippings. Since, by definition, a signal is minimal if and only if there does not exist a weaker admissible signal, the theorem follows.
\end{proof}
\begin{corollary}
For the constraint of at most $k$ consecutive dropouts, a signal is minimal if and only if every $1$ bit of the signal is surrounded by at least $k$ number of 
$0$ bits.
\end{corollary}
\begin{proof}
$1\rightarrow 0$ flipping of $1$ bit of the signal that is surrounded by at least $k$ number of $0$ bits is not admissible even for a combination of 
multiple $1\rightarrow 0$ flippings as it results in a signal with at least $k+1$ consecutive $0$s. The rest follows from Theorem \ref{flip_theo}.
\end{proof}
 \begin{figure}[ht]
\begin{center}
\includegraphics[scale=0.5]{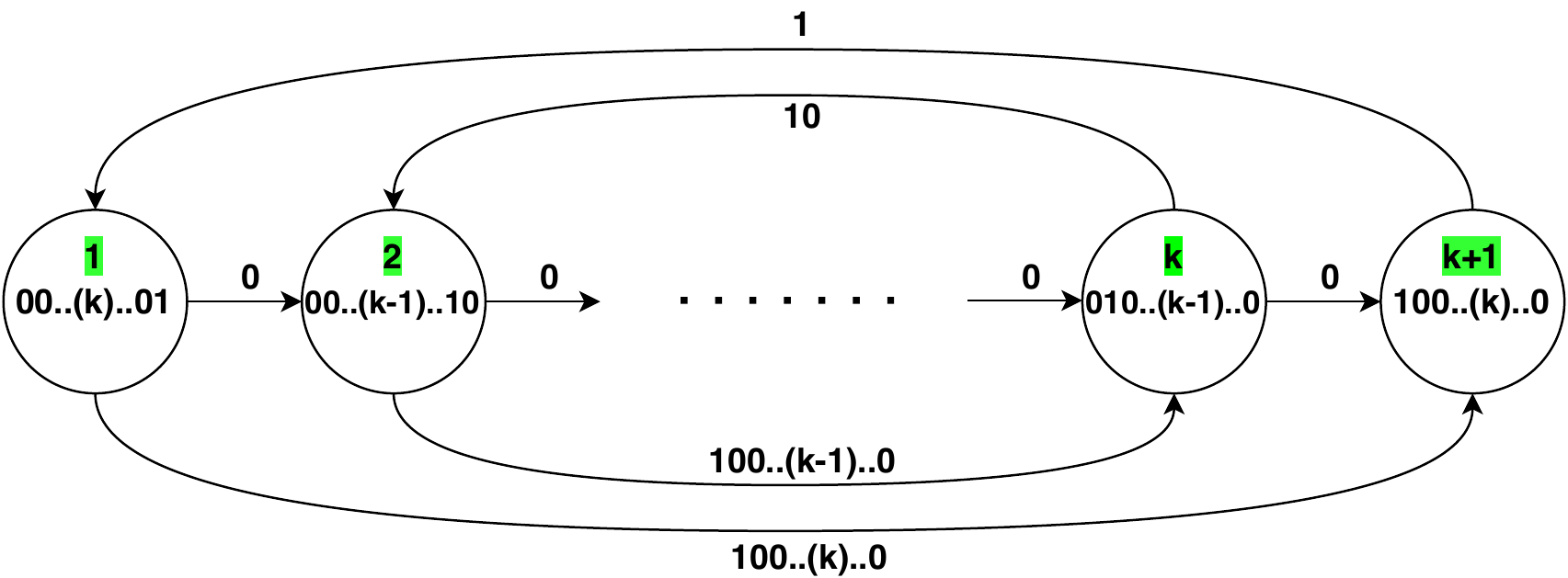}
\caption{Minimal Automaton for $k$ consecutive dropouts}
\label{fig:genminauto}
\end{center}
\end{figure}
The above result will now be used to construct an automaton (Fig. \ref{fig:genminauto}) whose paths are minimals signals for the constraint of at most $k$ 
consecutive dropouts. The automaton ensures that every $1$ bit is padded with enough $0$ bits such that it is surrounded by at least $k$ number of $0$ bits. 
Algorithm \ref{alg:minauto} generates paths of length $t$ of such an automaton in order to generate minimal signals of length $t$ for at most $k$ 
consecutive dropouts. It is basically a Breadth First Search algorithm searching for all signals of length $T$ on graph in Fig. \ref{fig:genminauto}. 
It maintains a queue of nodes to be evaluated, poping them one by one and pushing its neighboring nodes unless the length of the signal is equal to or 
exceeds $T$. All signals of length $T$ are stored and later used to evaluate performance of the optimal problems. The automaton in Fig. \ref{fig:genminauto} 
can be efficiently generated due to its symmetry; Every $i$th node, except the $(k+1)$th node, is connected to the $(i+1)$th node by a weight of $0$ and to 
the $(k+2-i)$th node by a weight of $1$ followed by $(k+1-i)$ zeros.

\begin{algorithm} 
 \caption{Minimal Signals using minimal automaton}
    \label{alg:minauto}
    \hspace*{\algorithmicindent} \textbf{Input : $k, T$}\\
    \hspace*{\algorithmicindent} \textbf{Output : Minimal signals $M_{\sigma(T)}$}
    \begin{algorithmic}[1]
    \Let{$M_{\sigma(T)}$}{$\emptyset$}
    \Let{$G(V,E)$}{k-minimal automaton (Fig. \ref{fig:genminauto})}
    \Let{$nodes$}{$Queue$}
    \State $nodes.push((1,\emptyset))$
    \While{$nodes$ is not empty}
    \Let{$(s,\sigma)$}{$nodes.pop$}
    \If{$|\sigma|=T$}
    \Let{$M_{\sigma(T)}$}{$M_{\sigma(T)}\cup\sigma$}
    \EndIf
    \If{$|\sigma|\geq T$}
    \Continue
    \EndIf
    \ForAll{ $e=(v,w,\hat{\sigma})\in E$ such that $v=s$}
    \State $nodes.push((w,\sigma\hat{\sigma}))$
    \EndFor
    \EndWhile
     \State \Return{$M_{\sigma(T)}$}
    \end{algorithmic}

\end{algorithm}
\begin{figure}[ht]
\begin{center}
\includegraphics[scale=0.5]{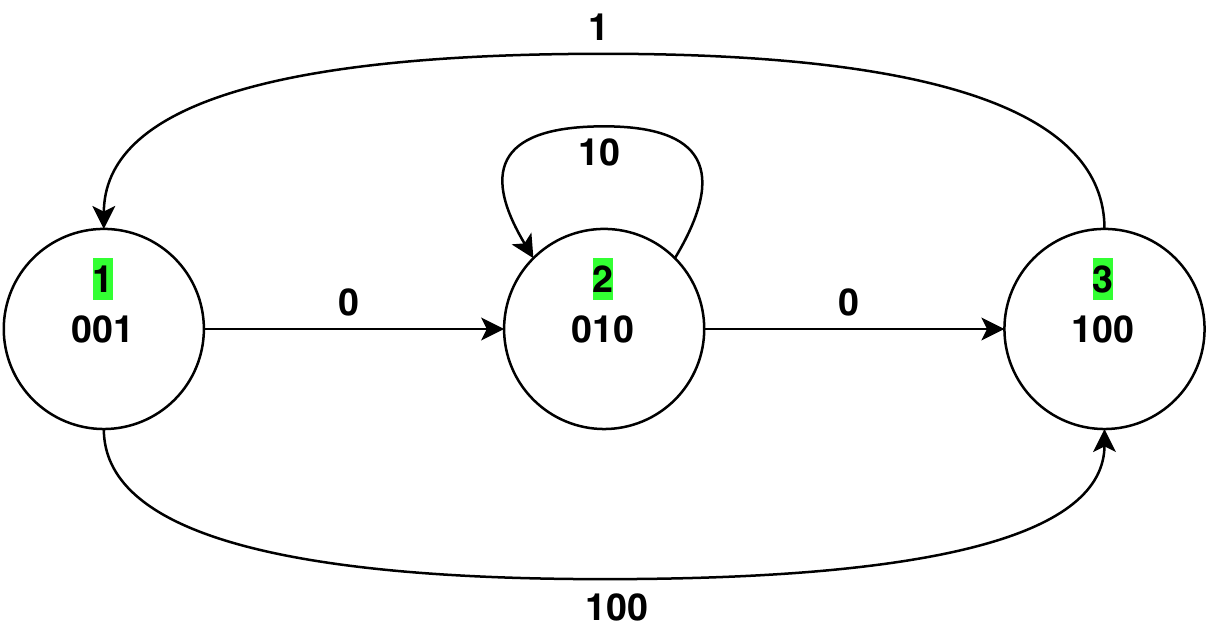}
\caption{Minimal Automaton for 2 consecutive dropouts}
\label{fig:2minauto}
\end{center}
\end{figure}

\begin{example}
Let's try to understand the working of minimal automaton by considering the example of at most $2$ consecutive dropouts (Fig. \ref{fig:2minauto}). 
Each state of the minimal automaton keeps track of the recent number of consecutive $0$ bits occurred in the switching signal. For at most $2$ consecutive 
dropouts, only $3$ states are possible: $1)$ No recent $0$ bits, $2)$ $1$ recent $0$ bit and $3)$ $2$ recent $0$ bits. For all except the $3$rd state both 
a $1$ and $0$ can occur next. However, for a signal in the $1$st state and the $2$nd state an occurrence of $1$ must be followed by two $0$ bits and one $0$ 
bit respectively in order to ensure minimality. For the $3$rd state, a $0$ is not allowed and only a $1$ is possible. Thus, any signal generated by the 
minimal automaton on any length is minimal and we are only required to generate signals of appropriate length which is precisely what Algorithm \ref{alg:minauto} does.
\end{example}
The following table (Table $\ref{tab:perf_analy}$) shows a comparison between the $T-$lift algorithm and the new algorithm in terms of 
the computation time in seconds to generate minimal signals of fixed length.
\begin{table}[h]
\centering
\caption{Computational Performance of algorithms}\label{tab:perf_analy}
\begin{tabularx}{0.85\linewidth}{|Y{1}|Y{1}|Y{1}|Y{1}|Y{1}|}
\hline
\multirow{3}{*}{Dropouts}&\multicolumn{4}{c|}{Computation Time(secs)}\\
\cline{2-5}
&\multicolumn{2}{c|}{Length 12}&\multicolumn{2}{c|}{Length 20}\\
\cline{2-5}
&Algo \ref{alg:minauto}&T-lift&Algo \ref{alg:minauto}&T-lift\\
\hline
1&0.0033&0.0051&0.0425&0.0931\\
\hline
2&0.0054&0.0089&0.1333&1.1134\\
\hline
3&0.0050&0.0123&0.1305&1.8862\\
\hline
\end{tabularx}
\end{table}
(The graph of time required to generate the minimal signals for a few fixed dropout constraints can be found in the supplementary material.)
\section{Validation Study}\label{study}
\begin{table}[h]
\centering
\caption{Results of Validation Study}\label{tab:results}
\begin{tabularx}{0.85\linewidth}{|Y{0.35}|Y{1.3}|Y{1.45}|Y{0.95}|Y{0.95}|}
\hline
Prob.&($k$, $n$, $m$, \#samples)&Avg. RPD&Avg. Time (Algo \ref{alg:minauto})&Avg. Time (T-lift)\\
\hline
\multirow{ 2}{*}{I}&(1,10,7,1200)&100\%&0.0019sec&0.0014sec\\
&(1,100,70,1200)&100\%&0.0237sec&0.0242sec\\
\hline
\multirow{ 2}{*}{II}&(1,10,7,765)&71\%&0.2011sec&0.2012sec\\
&(1,100,70,400)&35.83\%&0.5849sec&0.5855sec\\
\hline
\multirow{ 2}{*}{III}&(1,10,7,1200)&$4.2862\times 10^5$\%&0.0384sec&0.0405sec\\
&(1,100,70,1195)&9447.65\%&1.1651sec&1.1660sec\\
\hline
\multirow{ 2}{*}{V}&(1,10,7,1200)&$2.6892\times 10^7$\%&0.0360sec&0.0382sec\\
&(1,100,70,1200)&$1.1713\times 10^9$\%&0.5448sec&0.5459sec\\
\hline
\multirow{ 2}{*}{VI}&(1,10,7,1200)&$3.4899\times 10^{17}$\%&0.0018sec&0.0013sec\\
&(1,100,70,1200)&$2.4891\times 10^{22}$\%&0.0168sec&0.0151sec\\
\hline
\end{tabularx}
\end{table}
We now use Algorithm \ref{alg:minauto} and the $T-$lift algorithm used in \cite{packetdropsiccps} for studying the worst performance of optimally designed controllers 
and compare the performance of the two algorithms. 
Simulations performed to evaluate the worst possible impact of signal dropouts on optimal control problems described in Section \ref{pfm} using the 
proposed algorithms are given as follows. 
For each sample in the study, $A, B$ and $C$ matrices were generated randomly. The entries of $B$ and $C$ were 
derived from standard normal distribution. Samples of $A$ were obtained equally from the following $3$ methods: $(i)$ A diagonal matrix $D$ with random 
integers in $[-25,25]$ excluding $0$ was generated and scaled by $0.1$. A random orthogonal matrix $V$ was generated and a sample was generated by using  
$A=V^{\mathsf{T}}DV$; $(ii)$ Entries of $A$ were derived from the standard normal distribution; $(iii)$ Entries of $A$ were derived from standard the normal 
distribution and scaled by $10$. Non-singularity of $A$ was ensured in generating samples of $A$. Uncontrollable and unobservable systems were discarded. 
For control problems, $\mathbf{x_0}$ or $\mathbf{x_f}$ was chosen to be vectors of all ones. For each problem, two studies were performed: $(i)$ 
for systems with $10$ states and $7$ inputs/outputs; (ii) for systems with $100$ states and $70$ inputs/outputs. 
For each study, the computational performance (average running time) of Algorithm \ref{alg:minauto} and the T-lift algorithm of \cite{packetdropsiccps} and 
the relative performance degradation of each problem was evaluated and is provided in Table \ref{tab:results}. The relative performance degradation (RPD) 
due to packet dropouts is calculated as follows,
\begin{eqnarray}
\text{RPD}=\frac{\text{Worst Dropout Cost}-\text{No Dropout Cost}}{\text{No Dropout Cost}}.\nonumber
\end{eqnarray}
%

\section{Conclusion}
We considered time optimal control and estimation problems as well as energy/fuel optimization problems and LQR optimization 
problems for discrete linear systems subject to packet dropouts. 
We modeled the dropout signals using directed graphs and the dropout signals were binary in nature. Consequently, all the optimization problems were combinatorial. 
We tackled these problems algorithmically 
by exploiting the partial order relation imposed on the dropout signals. For a fixed dropout signal, the control problems 
reduce to SDP problems whereas; the estimation problem reduces to a feasibility problem. 
Therefore, the proposed optimization problems reduce to solving SDPs over the set of minimal signals. 
We gave a validation of the numerical experiments carried out using the proposed algorithms. We can use the solution of these optimization problems 
to quantify the reliability of the underlying communication network. This in turns allows us to compare different available communication networks and 
choose the optimal one based on the optimal control problem one wants to solve. 

In future, we want to consider probabilistic models for the dropout signals and uncertainties in system models and data measurements. Furthermore, we want to leverage model predictive control techniques 
to solve the optimization problems arising therein. 

 \bibliographystyle{ieeetr}   
\bibliography{dropouts}

\end{document}